\documentclass[11pt]{article}
\usepackage{fullpage}
\usepackage{framed}
\usepackage{xspace}
\usepackage{amsmath,amsthm,amssymb}
\newcommand{\remove}[1]{}

\newtheorem{thm}{Theorem}[section]
\newtheorem{claim}[thm]{Claim}

\newtheorem{lemma}[thm]{Lemma}
\newtheorem{dfn}[thm]{Definition}

\newtheorem{cor}[thm]{Corollary}

\def\_{\,\,\,\,\,}

\newcommand{\w}{w}

\newcommand{\B}{\{0,1\}}

\newcommand{\eps}{\epsilon}

\def\simpk{\textrm SIMPLE $\mathrm{k}$-PATH\xspace}

\def\minwt{{\textrm MIN-WT SIMPLE $\mathrm{k}$-PATH}\xspace}
\def\stminwt{{\textrm MIN-WT SIMPLE $\mathrm{(s,t)}$-$\mathrm{k}$-PATH}\xspace}
\def\ktsp{$\mathrm{k}$ \textrm{-TSP}\xspace}
\def\kstroll{$\mathrm{k}$-\textrm{STROLL}\xspace}
\def\ktour{$\mathrm{k}$-\textrm{TOUR}\xspace}

\def\kpath{{$\mathrm{k}$-path}\xspace}
\def\stpath{$\mathrm{(s,t)}$-\textrm{path}\xspace}

\def\stkpath{$\mathrm{(s,t)}$-$\mathrm{k}$-path\xspace}

\def\poly{n^{O(1)}}

\newtheorem{fact}[thm]{Fact}

\newcommand{\xorL}{ L_{\oplus}}
\newcommand{\F}{{\mathbb F}}

\newcommand{\size}{\mathrm{size}}

\newcommand{\set}[1]{\{#1\}}

\newcommand{\M}{{\cal M}}
\newcommand{\N}{{\cal N}}
\newcommand{\A}{{\cal A}}

\newcommand{\itab}[1]{\hspace{0em}\rlap{#1}}
\newcommand{\tab}[1]{\hspace{.12\textwidth}\rlap{#1}}
\newcommand{\htab}[1]{\hspace{.05\textwidth}\rlap{#1}}

\newcommand{\sett}[2]{\set{#1}_{#2}}
\title{Relations between automata and the simple $k$-path problem}
\author{Ran Ben-Basat\hspace{.3in}  Ariel Gabizon\thanks{The research leading to these results has received funding from the European Community's Seventh Framework Programme (FP7/2007-2013) under grant agreement number 257575.}\hspace{.3in} \\
 Technion, Haifa}
\begin{document}
\maketitle
\begin{abstract}Let $G$ be a directed graph on $n$ vertices.
A \kpath in $G$ is a path $p=v_1\to \ldots \to v_k$ in $G$.
Given an integer $k\leq n$, the \simpk problem asks whether there exists a \emph{simple} \kpath
 in $G$.
In case $G$ is weighted, the \minwt\;problem asks
for  a simple \kpath in $G$ of minimal weight.
The fastest currently known deterministic algorithm for \minwt by Fomin, Lokshtanov and Saurabh \cite{FLS13} runs in time $O(2.851^k\cdot \poly\cdot \log W)$ for graphs with integer weights in the range $[-W,W]$.
This is also the best currently known deterministic algorithm for \simpk- where the running time is the same without the $\log W$ factor.

We define $L_k(n)\subseteq [n]^k$ to be the set of words of length $k$
whose symbols are all distinct.
We show that an explicit construction of a non-deterministic automaton (NFA) of size $f(k)\cdot \poly$
for $L_k(n)$ implies an algorithm of running time $O(f(k)\cdot \poly\cdot \log W)$ for \minwt
when the weights are non-negative \emph{or} the constructed NFA is acyclic as a directed graph.
We show that the algorithm of Kneis et al. \cite{KMRR06} and its derandomization by Chen et al.\cite{CLSZ07} for \simpk can be used to construct an acylic NFA for $L_k(n)$ of size $O^*(4^{k+o(k)})$.

We show, on the other hand, that any NFA for $L_k(n)$ must have size at least $2^k$.
We thus propose closing this gap and determining the smallest NFA for $L_k(n)$ as an interesting open problem
that might lead to faster algorithms for \minwt.

We use a relation between \simpk  and \emph{non-deterministic xor automata} (NXA)
to give another direction for a deterministic algorithm with running time $O^*(2^k)$ for
\simpk.
\end{abstract}

\section{Introduction}
Let us recall the classic \emph{Travelling Salesman Problem} (TSP):
Given a complete undirected weighted graph $G$ on $n$ vertices
we wish to find a cycle of minimal weight passing through all vertices.
A parameterized version of this problem is sometimes called {\ktsp} (cf. \cite{BCKLMM07}). Here the salesman wants to visit only $k$ out of the $n$ cities (he does not insist on \emph{which} $k$)  while minimizing the total travel time.\footnote{There seem to be inconsistencies in the literature on whether \ktsp insists on
a cycle, or a fixed starting point.}
Note that in general the optimal route may \emph{not} be simple.
It will be convenient to formally define a more general problem, where the desired end point and
starting point are given as part of the input and the graph can be directed.
What we get is a problem referred to in \cite{BC13} as the \kstroll problem.

\begin{framed}
\noindent
\scriptsize
\textbf{\kstroll}\\
\normalsize

\itab{Input:} \tab{Directed graph $G=(V,E)$, vertices $s,t\in V$ , weight function $\w:E\to {\mathbb R}$.}

\itab{Parameter:} \tab{$k\in {\mathbb N}$.}

\itab{Problem:} \tab{Find a minimal weight path from $s$ to $t$ that visits at least $k$ distinct vertices}\\
 \tab{}\htab{}\tab{(counting $s$ and $t$).}
\end{framed}

The \ktour problem \cite{BC13} is a special case of \kstroll where $s=t$.
%
%
%
%
%

 A related problem that has received much attention is that of determining whether there
exists a \emph{simple} \kpath in a graph, and if
so returning such a path of minimal weight.
Here we define a \kpath in a graph to be a path of the form $v_1\to \ldots \to v_k$ (i.e., the number of \emph{vertices} in the path is $k$).
Let us define the unweighted and weighted versions of this problem.

\begin{framed}
\noindent
\scriptsize
\textbf{\simpk}\\
\normalsize

\itab{Input:} \tab{Directed graph $G=(V,E)$.}

\itab{Parameter:} \tab{$k\in {\mathbb N}$.}

\itab{Problem:} \tab{Determine if there exists a simple $k$-path in $G$, and if so return such a path.}
\end{framed}

\begin{framed}
\noindent
\scriptsize
\textbf{\minwt}\\
\normalsize

\itab{Input:} \tab{Directed graph $G=(V,E)$, weight function $w:E\to {\mathbb R}$.}

\itab{Parameter:} \tab{$k\in {\mathbb N}$.}

\itab{Problem:} \tab{Determine if there exists a simple $k$-path in $G$, and if so return such a path
of}\\\tab{}\htab{}\tab{minimal weight.}
\end{framed}

For vertices $s,t\in V$, an \stkpath is a \kpath beginning in $s$ and
ending in $t$.
Let us now also define \stminwt to be the version of \minwt where we give as additional
input vertices $s,t\in V$ and ask for a simple \stkpath of minimal weight.
Though an optimal solution for \kstroll is not necessarily a \emph{simple} path,
the problem is easily reducible to \stminwt:
Given $G$ compute the complete graph $G'$ on the same set of vertices, where the weight of the
directed edge $(u,v)$ is the weight of the minimal weight path between $u$ to $v$ in $G$.
A simple \stkpath $p'$ in $G'$ of minimal weight corresponds to an \stpath $p$ in $G$
passing through $k$ distinct vertices of minimal weight: Replace an edge $(u,v)$ in $p'$ by the shortest path
from $u$ to $v$ in $G$.
This connection gives more motivation for solving \minwt (it seems that all
known algorithms for \minwt can be adapted to solve \stminwt with the same running time).

The purpose of this paper is to propose a direction for obtaining faster deterministic algorithms
for \minwt via a connection to automata theory.
\subsection{Previous results on \minwt and our results}

Alon, Yuster and Zwick \cite{AYZ08} gave the first deterministic algorithm for \minwt running in time
$O(2^{O(k)}\cdot\poly\cdot \log W)$, where we assume the weights of the graph are integers
in the range $[-W,W]$.

The current state of the art is by Fomin, Lokshtanov and Saurabh \cite{FLS13}
giving a deterministic algorithm running in time  $O(2.851^k\cdot \poly\cdot \log W)$.

\begin{dfn}[The language $L_k(n)$]\label{dfn:L_k}
Fix positive integers $k\leq n$.
We define $L_k(n)\subseteq [n]^k$ to be the
set of words $w_1\cdots w_k \in [n]^k$ such
that $w_1,\ldots,w_k$ are all distinct.
\end{dfn}

Our main result is to show that a non-deterministic finite automaton (NFA) for the language
$L_k(n)$ implies an algorithm for \minwt whose running time is close to the size of the NFA.
In fact, Theorem \ref{thm-general_automat_alg} in Section \ref{sec:aut_algorithm} gives a general connection between
constructing compact NFAs and finding minimal-weight paths satisfying a certain constraint (in our case
the constraint is being simple of length $k$).

We state the result formally for \stminwt.
Note that \minwt can be easily reduced to \stminwt by
adding a start vertex $s$ that has outgoing edges to all vertices,
and a target vertex $t$ that has ingoing edges from all vertices.

The following theorem uses notation regarding NFAs from Definition \ref{dfn:NFA}.
We note in particular that by the \emph{size} of an NFA we mean
the total number of states \emph{and} transitions it contains.
\begin{thm}\label{thm:NFA_to_ksimpAlg}
Fix integers $k\leq n$. Suppose we can construct an NFA $M$ of size $s$
with $L(M)= L_k(n)$ in time $O(s)$.
Then we can solve \stminwt on graphs with $n$ vertices and non-negative
integer weights of size at most $W$ in time $O(s\cdot\log s\cdot  n^2\cdot \log W)$.

In case $M$ is a directed acyclic graph we can solve \stminwt on graphs with $n$ vertices
and integer weights in the range $[-W,W]$ in time
$O(s\cdot n^2\cdot  \log W)$.
\end{thm}

In Section \ref{sec:4^k} we show that the algorithms of Kneis et al. \cite{KMRR06} and  Chen et al.\cite{CLSZ07} for \simpk can be used to construct an acylic NFA for $L_k(n)$ of size $4^k\cdot k^{O(\log^2 k)}$ in time $O(4^k\cdot k^{O(\log^2 k)})$.
In Section \ref{section:lb} we show that any NFA for $L_k(n)$ must have at least $2^k$ states.
We thus find closing this gap to be an interesting problem that could lead to a faster deterministic
algorithm for \minwt.

A \emph{non-deterministic XOR automata} (NXA) is an NFA where the
acceptance condition is that a word has an \emph{odd} number of accepting paths, rather than
at least one.
In Section \ref{sec:XOR} we show that a small set of NXAs of size $O^*(2^k)$
can be constructed such that the union of their languages is $L_k(n)$.
This construction is in fact related to a randomized algorithm for \simpk
of Abasi and Bshouty \cite{AB13}.
We use this to give an $O^*(8^k)$ randomized algorithm for
\simpk.
The algorithm could be derandomized and its running time improved
potentially to $O^*(2^k)$ if a certain set of matrices
could be explicitly constructed and a faster algorithm for checking the
emptiness of an NXA were devised.
See Section \ref{sec:XOR} for details.
\section{Preliminaries}

We formally define non-deterministic automata.
It will be convenient to allow the transitions of the automaton to be weighted.
\begin{dfn}[NFA]\label{dfn:NFA}
A non-deterministic finite automaton (NFA) M over alphabet $\Sigma$ is a labeled directed graph
 $M = <Q,\Delta,q_0,F>$ where
 \begin{itemize}
  \item $Q$ is the set of vertices. We refer to the elements of $Q$ as `states'.
\item $\Delta$ is the set of edges. We refer to elements of $\Delta$ as `transitions' and suggestively use the notation $(u\to v)$ rather than $(u,v)$.
\item Each transition $e\in \Delta$ is labeled with an element of $\Sigma$.
\item $q_0$ is an element of $Q$ which is the `start state' of $M$.
\item $F\subseteq Q$ is the set of `accepting states'.
\end{itemize}
At times $M$ will be a \emph{weighted} graph.
That is, we will also have a weight function $\w:\Delta\to {\mathbb R}$.
For a word $w=w_1\cdots w_t \in [n]^t$, we define $M(w)\subseteq V$
to be the `subset of states reach by $w$' in the usual way for NFAs.
One subtlety: If while reading a word we reach a state where we cannot progress
by reading the next symbol, this run is rejected and the state we are at
is \emph{not} added to $M(w)$.
We define the \emph{language of $M$}, denoted $L(M)$ by
\[L(M) \triangleq \set{w\in [n]^*| M(w)\cap F \neq \emptyset}.\]

It will be convenient to define the \emph{size of $M$}, denoted $\size(M)$,
as the sum of the number of states and transitions in $M$.
That is, $\size(M)\triangleq |Q|+|\Delta|$.

Finally, we say $M$ is \emph{acyclic} if it is acyclic as a directed graph.
\end{dfn}

\begin{dfn}[Intersection NFA]\label{dfn:intersectionNFA}
Given NFAs $M_1=<Q_1,\Delta_1,q^1_0,F_1>$
and $M_2=<Q_2,\Delta_2,q^2_0,F_2>$ over the same alphabet $\Sigma$
we define the \emph{intersection NFA}
\[M_1\cap M_2\triangleq <Q_1\times Q_2,\Delta,<q^1_0,q^2_0>,F_1\times F_2>\]
over $\Sigma$, where the set of transitions $\Delta$ is defined as follows.
For every pair of transitions $(u_1\to v_1)\in \Delta_1$ and $(u_2\to v_2)\in \Delta_2$
that are \emph{both labeled by the same element $a\in \Sigma$},
we have a transition $(<u_1,u_2>\to <v_1,v_2>)\in \Delta$ labeled $a$.
\end{dfn}
\noindent
It is known that
\begin{fact}\label{fact:intersectionNFA}
$L(M_1\cap M_2) = L(M_1)\cap L(M_2)$.
\end{fact}

\section{Finding automata-constrained shortest paths}\label{sec:aut_algorithm}
The purpose of this section is to establish a general connection between
algorithms for finding minimal weight paths satisfying a certain constraint
and NFAs representing the constraint.

The following definition and straightforward lemma formally convert a graph
into an automaton accepting the paths of the graph.
\begin{dfn}[The path automaton]\label{dfn:path_aut}
Let $G=<V,E>$ be a directed graph.
Fix $s,t \in V$.

The NFA
 \[M(G,s,t)\triangleq <Q=V\cup\set {q_0},\Delta= E\cup\set{(q_0,s)},q_0, F=\set{t}>\]
with alphabet $\Sigma = V$ is defined with the following labeling of transitions.
The transition $(q_0 \to s)$ will be labeled $s$. For each $(u,v)\in E$ the  transition $(u \to v)$ is labeled with the source vertex $u\in V$ of the edge.
\end{dfn}
\begin{lemma}
Let $G=<V,E>$ be a directed graph.
Fix $s,t \in V$.
Then $L(M(G,s,t))$ is precisely the set of words $p=s\cdot v_1\cdots v_m\cdot t$
such that $s\to v_1\to \ldots\to v_m \to t$ is a path from $s$ to $t$ in $G$.
\end{lemma}

\begin{dfn}[Paths accepted by an NFA]\label{dfn-accepted_paths}
Fix an NFA $M$ with alphabet $\Sigma $,
and a directed graph $G=<\Sigma,E>$.
Let $p=v_1\to v_2\to \ldots..\to v_t$ be a (directed) path in $G$.
Identify $p$ with the word $v_1\cdots v_t\in \Sigma^t$.
We say the path $p$ is \emph{accepted by $\mathrm{M}$} if $p\in L(M)$.
Or in words, running the NFA $M$ with the word $p$ can end in an accepting state.
\end{dfn}

The following theorem states that if we have an NFA of a certain size
capturing a certain constraint on a path, we have
an algorithm for finding the shortest path satisfying the constraint
whose running time is similar to the size of the NFA.
\begin{thm}\label{thm-general_automat_alg}
Fix any NFA $M$ with alphabet $\Sigma=[n]$.
There is an algorithm that, given as input a directed weighted graph $G=<[n],E,\w>$ with integer weights and vertices $s,t \in [n]$,
returns an \stpath in $G$ that is accepted by $M$ of minimal weight.
The running time of the algorithm is at most $O(\size^2(M) \cdot n^3)$.
The running time can be improved to
\begin{itemize}
\item  $O(\size(M)\cdot \log(\size(M)) \cdot n^2)$ when $G$ only contains non-negative weights.
\item $O(\size(M)\cdot n^2)$ when $M$ is acylic.
\end{itemize}
All running times assume $O(1)$ arithmetic operations on weights
and without this assumption require an additional $\log W$ factor
when the weights are in the range $[-W,W]$.
\end{thm}
\begin{proof}
First note that we can convert $M$ to an NFA with one accepting state
while at most doubling its size. Let us assume from now on that $M$ indeed has a unique accepting state.
Let $M(G,s,t)$ be the NFA from Definition \ref{dfn:path_aut}.
We construct the intersection NFA $N\triangleq M\cap M(G,s,t)$ as in Definition \ref{dfn:intersectionNFA}.
Now we add weights to the transitions according to $G$.
More precisely, transitions $(<q_1,u>\to <q_2,v>)\in \Delta$ with $(u,v)\in E$ will be given
weight $\w(u,v)$.
All other transitions (simply ones where the second coordinate shifts from the start state of $M(G,s,t)$ to $s$)
will be given weight $0$.

Note that when looking at $N$ as a weighted directed graph, the paths from its start to accept state
exactly correspond to the $\mathrm{(s,t)}$-paths in $G$ accepted by $M$. (It is possible that a certain
\stpath in $G$ corresponds to many accepting paths in $N$).
Now note that the weight of any accepting path in $N$ of a word $s\cdot v_1\cdots v_m\cdot  t$ is
the same as the weight of the path $s\to v_1\to \ldots\to v_m\to t$ in $G$.
Thus, running a shortest path algorithm on $N$ from the start to accept state
will give us an \stpath in $G$ that is accepted by $M$ and is of minimal weight among the $\mathrm{(s,t)}$-paths in $G$ accepted by $M$.
Note that $N$ has at most $\size(M)\cdot (n+1)$ vertices
and at most $\size(M)\cdot n^2$ edges.
Running the Bellman-Ford algorithm would give us
time $O(|V|\cdot |E|) =O(\size^2(M)\cdot n^3)$.
In case $G$ has non-negative weights we can use Fredman and Tarjan's implementation of Dijkstra's algorithm
\cite{FT87}
to get time $O(|E|+ |V|\cdot \log |V|) = O(\size(M)\cdot \log (\size(M))\cdot n^2)$.
In case $M$ is acyclic so is $N$ and we can use toplogical sort to
get time $O(\size(M)\cdot n^2)$.
\end{proof}

\noindent Theorem \ref{thm:NFA_to_ksimpAlg} now follows from Theorem \ref{thm-general_automat_alg}
by considering NFAs whose language is $L_k(n)$.

\section{A Lower bound for the NFA size of $L_k(n)$}\label{section:lb}
The following theorem of Gliaster and Shallit \cite{GS96} gives a method to lower bound
the NFA size of a language.
\begin{thm}\label{thm:NFAlb_method}
Fix a language $L\subseteq [n]^*$.
Suppose we have elements $x_1,\ldots,x_t,y_1,\ldots,y_t \in [n]^*$
such that
\begin{itemize}
\item For all $i\in [t]$, $x_i\cdot y_i \in L$.
\item For all $i\neq j\in [t]$, $x_i\cdot y_j\notin L$.
\end{itemize}
Then any NFA for $L$ has at least $t$ states.
\end{thm}

\begin{thm}\label{thm:lb_for_Lk}
Fix any integers $k\leq n$.
Then any NFA for $L_k(n)$ has at
least $2^k$ states.
\end{thm}
\begin{proof}
For every subset $S=\set{i_1,\ldots,i_d}\subseteq [k]$,
let $x_S\in [n]^k$ be the word $x_S=i_1\cdots i_d$.
For every $S\subseteq [k]$ define $y_S = x_{\bar{S}}$.
It is clear that for every $S\subseteq [k]$,
$x_S\cdot y_S \in L$.
And for every $S\neq T \subseteq [k]$
$x_S\cdot y_T \notin L$.
Now the claim follows from Theorem \ref{thm:NFAlb_method}.
\end{proof}


\section{NFA construction for $L_k(n)$}\label{sec:4^k}

In this section we give an explicit construction of an NFA for
the language $L_k(n)$ of size $O^*(4^{k+o(k)})$.
The NFA construction and analysis closely correspond to the algorithm
for \simpk of \cite{KMRR06} and its derandomization
using universal sets by \cite{CLSZ07}.
For this purpose we now define universal sets.

\begin{dfn}[$(n,k)$-universal set] \label{dfn:universal_set}
A set of strings $U\subseteq \B^n$ is an \emph{$(n,k)$-universal set} if
for every $S\subseteq [n]$ of size $k$, and every $a\in \B^k$ we have $x\in T$
such that $x|_S=a$.
Equivalently, an $(n,k)$-universal set is a set $U$ of subsets of $[n]$
such that for every $S\subseteq [n]$ of size $k$ and every $S'\subseteq S$
we have $T\in U$ such that $T\cap S = S'$.
\end{dfn}


\noindent
Naor, Schulman and Srinivasan \cite{NSS95} gave an almost
optimal construction of universal sets.

\begin{claim}\label{clm:univsetexplicit}[\cite{NSS95}]
Fix integers $k\leq n$. There is a deterministic algorithm
of running time $O(2^k\cdot k^{O(\log k)}\cdot \log n )$ that constructs
an $(n, k)$-universal set of size $2^k\cdot k^{O(\log k)} \cdot \log n$.
\end{claim}

\noindent We now state the main result of this section.
\begin{thm}\label{thm:4k_aut}
Fix integers $k\leq n$. An acyclic NFA $M$ of size $O^*(4^k\cdot k^{O(\log^2 k)})$
for $L_k(n)$ can be constructed in time $O^*(4^k\cdot k^{O(\log^2 k)})$.

\end{thm}

\noindent Before proving the theorem we state a technical claim that will be used in the analysis.
\begin{claim}\label{clm:ceil}
For a positive integer $k$ look at the sum
\[s(k) = k+\lceil k/2 \rceil + \lceil \lceil k/2 \rceil /2 \rceil + \ldots + 1.\]
Then
\begin{itemize}
\item $s(k)\leq 2k+ 2\cdot \log k$
\item The number of summands in $s(k)$ is at most $\log k + 1$.
\end{itemize}
\end{claim}
\noindent We proceed with the proof of Theorem \ref{thm:4k_aut}.
\begin{proof}
The following definition will be convenient for the proof.
For a subset $S\subseteq [n]$ we define the language $L_k(n,S)\triangleq L_k(n)\cap S^k$.
In words, $L_k(n,S)$ is simply the set of words in $w\in [n]^k$ whose symbols are all distinct,
and are also all in $S$.
Fix any positive integer $n$.
For every $1 \leq k \leq n$ and $S\subseteq [n]$ we construct an NFA $M_{k,S}$ for $L_k(n,S)$
by induction on $k$ as follows.

For $k=1$, given $w\in [n]^k$ the $M_{k,S}$ will simply check
if $w_1\in S$ and if $|w|=1$. Such $M_{k,S}$ of size $3\cdot n$ can be constructed.
Now assume we have a construction of an NFA $M_{k',S}$ for every $1\leq k' < k$
and $S\subseteq[n]$.
Before constructing $M_{k,S}$, let us construct as a component
 an NFA for a simpler language.
Fix disjoint subsets $S_1, S_2\subseteq [n]$.
We will define an NFA $M_{k,S_1,S_2}$
that accepts exactly the words $w\in L_k(n)$
whose first $\lceil k/2 \rceil$ symbols are in $S_1$, and last
$\lfloor k/2 \rfloor$ symbols are in $S_2$.
$M_{k,S_1,S_2}$ can be constructed as follows.
$M_{k,S_1,S_2}$ will consist of a copy
of $M_{\lceil k/2 \rceil,S_1}$ that reads the first $\lceil k/2 \rceil$
symbols of $w$, followed by a copy of $M_{\lfloor k/2 \rfloor,S_2}$
that reads the last $\lfloor k/2 \rfloor$ symbols of $w$.

Now, given $S\subseteq [n]$ we construct $M_{k,S}$ as follows.
Fix an $(n,k)$-universal set $U$ of size $|U|= 2^k\cdot k^{O(\log k)}\cdot \log n$
obtained from Theorem \ref{clm:univsetexplicit}.
For every set $T\in U$ we put an  $\eps$-transition from the start state of $M_{k,S}$
to a copy of the NFA $M_{k, S\cap T,S\cap \bar{T}}$.
Thus, $M_{k,S}$ accepts a word $w$ if and only if one of the automata $\sett{M_{k, S\cap T,S\cap \bar{T}}}{T\in U}$
accepts $w$.
Let us show that indeed $L(M_{k,S}) = L_k(n,S)$.
Note that for any disjoint subsets $S_1,S_2\subseteq S$,
$M_{k,S_1,S_2}$ accepts a \emph{subset} of $L_k(n,S)$. Hence,
it is clear that $M_{k,S}$ does not accept any words outside of $L_k(n,S)$.
Now, fix a word $w\in L_k(n,S)$ and let us show that one the machines $\sett{M_{k, S\cap T,S\cap \bar{T}}}{T\in U}$
accepts it.
Let $S_1\subseteq S$ be the set of the first $\lceil k/2\rceil$ symbols that appear in $w$.
Let $S_2\subseteq S$ be the set of the last $\lfloor k/2\rfloor$ symbols that appear in $w$.
Note that as $w\in L_k(n,S)$, $S_1$ and $S_2$ must be disjoint and $|S_1\cup S_2|=k$.
From the property of an $(n,k)$-universal set, there must exists a set $T\in U$
such that $T\cap(S_1\cup S_2) = S_1$.
For this $T$ $M_{k, S\cap T,S\cap \bar{T}}$ accepts $w$.
We have shown that $L(M_{k,S}) = L_k(n,S)$.
Now let us bound the size of $M_{k,S}$.
For $k\leq n$, denote by $T_k$ the maximum over $S\subseteq [n]$ of the size of the NFA $M_{k,S}$
constructed in this way.
Using this notation we have for any disjoint subsets $S_1, S_2\subseteq [n]$ that
\[|M_{k,S_1,S_2}|\leq T_{\lceil k/2\rceil} + T_{\lfloor k/2\rfloor} + 1 \leq 2\cdot  T_{\lceil k/2\rceil} + 1\]
, where $|M_{k,S_1,S_2}|$ denotes the size of $M_{k,S_1,S_2}$ in the construction described above.
Now note that
$M_{k,S}$ consists of $|U|$ copies of machines $M_{k,S_1,S_2}$ (and the $\eps$-transitions to these copies).
Using this we have
\[T_k \leq  2^k\cdot k^{O(\log k)}\cdot \log n \cdot 2\cdot ( T_{\lceil k/2\rceil} + 1) +  2^k\cdot k^{O(\log k)}\cdot \log n +1
=2^k\cdot k^{O(\log k)}\cdot \log n \cdot T_{\lceil k/2\rceil}.\]
Using Claim \ref{clm:univsetexplicit}, and $T_1\leq 3n$ we get
\[ T_k \leq 2^{2k + 2\log k}\cdot k^{O(\log^2 k)}\cdot \log n ^{\log k + 1}\cdot 3n\]
Using the fact that for any $k$, either $\log n ^{\log k } \leq k^{\log^2 k}$  or
$\log n ^{\log k} \leq n$ we can write
\[T_k = O^*(4^k\cdot k^{O(\log^2 k)}).\]
\end{proof}

\section{Non-deterministic XOR automata for $L_k(n)$}\label{sec:XOR}
Informally, a non-deterministic xor automaton (NXA) is simply an NFA
where the acceptance criteria for a word is that \emph{there is an odd number of accepting paths for $w$},
rather than just one.
It will be convenient to formally define the \emph{XOR-language} of an NFA rather than
formally defining NXAs.
\begin{dfn}[The language $\xorL$]\label{dfn:mult_auto}
Let $M$ be a non-deterministic xor automaton over an alphabet $\Sigma$. We define the \emph{XOR-language}
of $M$, denoted $\xorL(M)\subseteq \Sigma^*$, to be the set of words $w$ that have an odd number of paths to
an accept state in $M$.
\end{dfn}

The purpose of this section is to construct a small set of NFAs of size $O(2^k\cdot k\cdot n)$ such that the union of their
XOR-languages is $L_k(n)$. This construction can be viewed as an `automata interpretation' of (a simplified version) of the algorithm for \simpk of Abasi and Bshouty \cite{AB13}. This will be used to get an algorithm for \simpk with running time $O^*(8^k)$. We proceed with the construction.

 In the rest of this section sums are always in $\F_2$, i.e., modulu $2$. For each non-empty subset $S\subseteq [k]$, define the function $\phi_S:(\B^k)^k\to \B$ by
\[\phi_S(v_1,\ldots,v_k) \triangleq \prod_{i=1}^k \sum_{j\in S} v_{i,j}\]

and define $\phi:(\B^k)^k \to \B$ by
\[\phi(v_1,\ldots,v_k) \triangleq \sum_{\emptyset\neq S\subseteq [k]} \phi_S(v_1,\ldots,v_k).\]

From Ryser's formula for the permanent\cite{R63}  we know that
\begin{lemma}\label{lem:phi_is_det}
 $\phi(v_1,\ldots,v_k)$ is equal to the determinant
of the $k\times k$ matrix over $\F_2$ whose columns are $v_1,\ldots,v_k$.
\end{lemma}

Fix a $k\times n$ matrix $A$ over $\F_2$ with columns $v_1,\ldots,v_n \in \B^k$.
For each non-empty subset $S\subseteq [k]$,
we define a function $f_{A,S}:[n]^k\to\B$
by $f_{A,S}(i_1,\ldots,i_k) \triangleq \phi_S(v_{i_1},\ldots,v_{i_k})$.
We define $f_A:[n]^k\to \B$ by
\[f_A(i_1,\ldots,i_k) \triangleq \phi(v_{i_1},\ldots,v_{i_k}) = \sum_{\emptyset\neq S\subseteq [k]} \phi_S(v_{i_1},\ldots,v_{i_k})=\sum_{\emptyset\neq S\subseteq [k]} f_{A,S}(i_1,\ldots,i_k).\]

\begin{lemma}\label{lem:aut_for_phiS}
Fix any $k\times n$ matrix $A$ over $\F_2$ and non-empty $S\subseteq [k]$.
There is a deterministic automaton $M_{A,S}$ for $f_{A,S}^{-1}(1)$ with
$k+1$ states and at most $k\cdot n$ edges.
\end{lemma}
\begin{proof}
Let $v_1,\ldots,v_n$ be the columns of $A$.
Let $T\subseteq [n]$ be the set of elements $i\in [n]$ such that
\[\sum_{j\in S} v_{i,j}=1.\]
Observe that $f_{A,S}(i_1,\ldots,i_k)=1$ if and only if $i_1,\ldots,i_k$ are all contained in $T$.
This motivates the following construction:
$M_{A,S}$ will contain the start state $q_0$,
and the states $q_1,\ldots,q_k$.
$q_k$ will be the only accept state.
For each $0\leq j \leq k-1$, and for every $i\in S$.
There will be an edge from $q_j$ to $q_{j+1}$ labeled $i$.
\end{proof}

\begin{thm}\label{lem:aut_for_phi}
Fix any positive integers $k\leq n$ and any $k\times n$ matrix $A$ over $\F_2$.
There is an NFA $M_A$ over $[n]$ of size $O(2^k\cdot k\cdot n)$ such that
$\xorL(M_A) = f^{-1}(A)$.
\end{thm}
\begin{proof}
For every non-empty $S\subseteq [k]$,
$M_A$ will contain a copy of the automaton $M_{A,S}$ as described in Lemma \ref{lem:aut_for_phiS}.
We unite the start state $q_0$ and accept state $q_k$ of all the automata $M_{A,S}$
to one start state $q_0$ and accept state $q_k$ of $M_A$.
$\xorL(M_A)$ contains exactly the words $(i_1,\ldots,i_k)$ that
are accepted by an odd number of the automata $M_{A,S}$.
Since $L(M_{A,S}) = f^{-1}_{A,S}(1)$, this is exactly $f^{-1}_A(1)$.
\end{proof}

\subsection{Covering matrices}
We wish to show there is a small set of matrices $A$ such that
the union of the XOR-languages of the corresponding automata $M_A$
is equal to $L_k(n)$.
This motivates the following definition.
\begin{dfn}\label{dfn:covering}
Let $\A$ be a set of $k\times n $ matrices over $\F_2$.
We say $\A$ is \emph{$(n,k)$-covering}, if for every subset of $k$ distinct columns
$I=(i_1,\ldots,i_k)\subseteq [n]$, there is a matrix $A\in \A$
such that the columns $(i_1,\ldots,i_k)$ in $A$ are linearly independent.
\end{dfn}

From now on for $I=(i_1,\ldots,i_k)\subseteq [n]$ and a $k\times n $ matrix $A$ over $\F_2$
we denote by $A_I$ the restriction of $A$ to the columns $(i_1,\ldots,i_k)$.
\begin{lemma}\label{lem:covering_exist}
Fix any positive integers $k\leq n$.
There exists a set $\A$ of $k\times n $ matrices over $\F_2$
that is \emph{$(n,k)$-covering} with $|\A|\leq 2k\cdot \log n$.
\end{lemma}
\begin{proof}
We use the probabilistic method.
It is known that when choosing a random $k\times k$ matrix over $\F_2$
the probability that it is non-singular is at least half.
Fix $I=(i_1,\ldots,i_k)\subseteq [n]$.
It follows that when choosing a random $k\times n$ $A$ matrix over $\F_2$,
the probability that $A_I$ is singular is at most half.
Thus, when independently choosing $2k\cdot \log n$ random $k\times n$ matrices $A^1,\ldots, A^{2k\cdot \log n}$
the probability that the columns $I$ are dependent in all of them
is at most $2^{-2k\cdot \log n}= n^{-2k}$.
Taking a union bound over all $\binom{n}{k}\leq n^k$ choices of $I$
we see there must be a choice of $\A=\set{A^1,\ldots, A^{2k\cdot \log n}}$ that is $(n,k)$-covering.
\end{proof}

\begin{thm}\label{thm:covering_gives_L_k}
Fix any positive integers $k\leq n$.
Let $\A$ be a family of $k\times n $ matrices over $\F_2$.
that is \emph{$(n,k)$-covering}.
Then the union of languages $\bigcup_{A\in \A} \xorL(M_A)$
is equal to $L_k(n)$.
\end{thm}
\begin{proof}
Fix a word $w=(i_1,\ldots,i_k)\in [n]^k$ that
is not in $L_k(n)$.
Then for any $k\times n$ matrix $A$, $f_A(i_1,\ldots,i_k)$ is equal
to the determinant of a $k\times k$ matrix that has at least two identical columns
so $f_A(i_1,\ldots,i_k)=0$. This exactly means that $w\notin \xorL(M_A)$.
On the other hand, given $w=(i_1,\ldots,i_k)\in L_k(n)$, i.e. $i_1\neq \ldots\neq i_k$, we have some $A\in \A$
such that the columns $(i_1,\ldots,i_k)$ in $A$ are linearly independent.
For this $A$, $f_A(i_1,\ldots,i_k)=1$ and therefore $w\in \xorL(M_A)$.
\end{proof}
\begin{cor}\label{cor:mult_aut_for_Lk}
Fix any positive integers $k\leq n$.
There is a $\M$  of  $ 2k\cdot \log n$ NFAs, each of size at most $O(n\cdot 2^k)$ such that
 $\bigcup_{M\in \M} \xorL(M)= L_k(n)$.
\end{cor}
\subsection{An algorithm for \simpk via XOR automata}
We now construct an NFA whose XOR-language is the
set of simple $k$-paths in a graph.

\begin{cor}\label{cor:xor_aut_for_ksimp}
Fix any positive integers $k\leq n$.
Fix a directed graph $G=<[n],E>$.
Fix vertices $s,t\in V$.
There is a set $\N$ of  $ 2k\cdot \log n$ NFAs, each of size at most $O^*(2^k)$ such that
 $\bigcup_{N\in \N} \xorL(N)$ is exactly the set of simple $\mathrm{(s,t)}$-$k$-paths in $G$.
\end{cor}
\begin{proof}

We take the family $\M$ of $NFA$'s from Corollary \ref{cor:mult_aut_for_Lk}.
For each $M\in \M$ we compute the intersection NFA $N= M\cap M(G,s,t)$.

Note that the number of accepting paths of a word $w$ in  $N$ is the product
of the number of accepting paths in $M$ and $M(G,s,t)$.
As $M(G,s,t)$ is deterministic, this means $\xorL(N)$ is exactly the set of words
in $\xorL(M)$ that are also $\mathrm{(s,t)}$-paths in $G$.
We take $\N$ to be the set of all these NFAs $N$.
Hence $\bigcup_{N\in \N} \xorL(N)$ is the intersection of $L_k(n)$
with the  set of $\mathrm{(s,t)}$-$k$-paths in $G$.
\end{proof}

The work of Vuillemin and Gama \cite{VG09} on minimizing NXA gives in particular a method
to check if the XOR-language of an NFA is empty.

\begin{thm}[\cite{VG09}]\label{thm:NXA_check_empty}
There is a deterministic algorithm, that given an NFA $M$
with $s$ states, checks in time $O(s^3)$ whether $\xorL(M)=\emptyset$.
\end{thm}

Given a set of $\A$ of $(n,k)$-covering matrices with we could now use Theorem \ref{thm:NXA_check_empty} to solve \simpk in deterministic time $|\A|\cdot O^*(8^k)$.
However, currently there are no explicit constructions of such sets with $|\A|=\poly$.
The only explicit construction we are aware of is implicit in Lemma 51 of Bshouty\cite{Btesters}
and gives $|\A| = 2^{O(k)}\cdot \log n$.
Choosing $|\A|$ randomly would lead to a randomized algorithm for \simpk with running
time $O^*(8^k)$.
We state two open problems whose solution could lead to an $O^*(2^k)$ deterministic algorithm
for \simpk.
\begin{cor}
Suppose that
\begin{itemize}
\item Given integers $k\leq n$ we can construct a set $\A$ of $(n,k)$-covering matrices
in time $\poly$ with $\A = \poly$.
\item Given an NFA $N$ we can check in deterministic time $O(\size(N))$
whether $\xorL(N) = \emptyset$.
\end{itemize}
Then we can solve \simpk deterministically in time $O^*(2^k)$.
\end{cor}

\section*{Acknowledgements}
We thank Hasan Abasi, Nader Bshouty, Michael Forbes and Amir Shpilka for helpful conversations.
\bibliography{ktsp}
\bibliographystyle{plain}

\end{document}